\documentclass[pra,amsmath,amsfonts,twocolumn,superscriptaddress,showpacs,showkeys,article]{revtex4-1}
 \usepackage{epsf,epsfig}
 \usepackage[psamsfonts]{amssymb}
\usepackage{amsthm}
\usepackage{graphicx}
\usepackage{color,soul}
\usepackage{xcolor}
\usepackage{bbold}
\usepackage{hyperref}
\usepackage{csquotes}
\usepackage{subfigure}

\newtheorem{theorem}{Theorem}
\newtheorem{lemma}[theorem]{Lemma}

 \newcommand{\Id}{{\mathbb 1}}
 
 \newcommand{\rank}{{\mathrm rank}}
 
 \newcommand{\T}{{\mathrm t}}
 \newcommand{\can}{{\mathrm can}}

\newcommand{\Tr}{{\mathrm {Tr}}}
\newcommand{\diag}{{\mathrm {diag}}}

\begin{document}
\preprint{APS/123-QED}

\title{Measuring quantum discord using  the most distinguishable steered states  }

\author{Vahid Nassajpour}
\author{Seyed Javad Akhtarshenas}
\email{akhtarshenas@um.ac.ir}
\affiliation{Department of Physics,  Ferdowsi University of Mashhad, Mashhad, Iran}

\begin{abstract}
Any two-qubit state can be  represented, geometrically, as an ellipsoid  with a certain  size and a center  located within the Bloch sphere of one of the qubits. Points of this ellipsoid represent the post-measurement states when  the other qubit  is measured.  Based on the most demolition concept in the definition of quantum discord, we  study the amount of demolition when the two post-measurement states, represented as two points on the steering ellipsoid, have the most  distinguishability.   We use trace distance as a measure of distinguishability and obtain the maximum distinguishability for some classes of states, analytically. Using the optimum measurement that gives the most distinguishable steered states, we extract quantum correlation of the state and compare the result  with the quantum discord. It is shown that  there are some important classes of states for which the most demolition happens exactly at the most distinguished steered points. Correlations gathered from the most distinguished post-measurement states provide  a faithful and  tight upper bound touching the quantum discord in most of the cases.

\end{abstract}


\maketitle
\section{Introduction}

In a   bipartite quantum system containing   some kind of correlations,  when one side is measured locally, the state of the other side may be collapsed to some specified states. It means that one side, say Bob's side,  can steer the state of the  other side, say Alice's side,  just by performing local measurement on his particle. This notion of quantum steering, introduced by Schr\"{o}dinger \cite{Schrodinger1935,Schrodinger1936}, is closely related  to the  concept of EPR nonlocality \cite{EPR1935}.  The states to which Alice's particle steers to can be specified by the basis on which Bob performs  measurement on his particle.   Considering  all positive operator  valued measures (POVMs), Bob can steer  the Alice's particle  to a set of post-measurement states. In the case of  two-qubit systems, this set of post-measurement states forms an   ellipsoid, i.e. the so-called quantum steering ellipsoid (QSE),  living in the Alice's Bloch sphere \cite{MilnePRL2014-1}. This ellipsoid  is unique up to the  local unitary transformations for  any two-qubit system \cite{MilnePRL2014-1}. Having this  geometry, it is useful to study some non-classical features of composite systems such as entanglement, separability, negativity, fully entangled fraction, quantum discord, Bell non-locality, monogamy, EPR steering and even the dynamic of a quantum system \cite{MilnePRL2014-1,MilneNJP2014-2,MilneNJP2015,MilnePRA2014,HuPRA2015,ShiNJP2011,WangPRA2014,CloskeyPRA2017,Schrodinger1935,WisemanPRA2007,
WisemanPRL2007,CavalcantiPRL2014}. When the results of the measurement are not recorded, the  measurement  performed locally by Bob cannot affect the Alice's reduced density matrix. Therefore the ensemble average of the Alice's Bloch vectors of the    post-measurement states,  produced by a set of POVM on the Bob's part,  must be equal to the  coherence vector of the Alice's reduced state, meaning that  the coherence vector lies inside the ellipsoid.   In particular,  when Bob's reduced state  is totally mixed, the Alice's coherence vector coincides on her ellipsoid center \cite{MilnePRL2014-1}. Such a state is called \enquote{canonical state}.

Conversely, we can reconstruct a two-qubit state from its ellipsoid, given the coherence vectors of two parts \cite{MilnePRL2014-1}. However, not any ellipsoid can belong to a physical state. For example, any physical ellipsoid touches the Bloch sphere  at most at two points unless it is the whole  Bloch sphere \cite{BraunJPMT2014}. Given the ellipsoid center, authors in \cite{MilneNJP2014-2} have been studied conditions of physicality and separability of canonical states. Based on the Peres-Horodecki criterion \cite{PresPRL1996,HorodeckiPL1996}, the authors of  \cite{MilnePRL2014-1} have shown   that the separability of  the canonical states  depends on the shape of their ellipsoids.

All of the above symmetric features can be observed from the Bob's ellipsoid which its dimension  is the same as the Alice's one \cite{MilnePRL2014-1}. Quantum discord (QD) \cite{ZurekPRL2002,VedralJPA2001}  is an asymmetric measure of quantum correlations that could be obtained by eliminating the classical correlation from the total correlation, measured by the mutual information, by means of the most destructive measurement on the one party of the system (for a review on quantum discord see \cite{ModiRMP2012}).  The total  information shared between  parts of a bipartite  quantum state $\rho$ is given by
\begin{equation}
I({\rho }) = S({\rho ^A}) + S({\rho ^B}) - S({\rho ^{AB}}),
\label{10}
\end{equation}
where $\rho^A=\Tr_B{\rho}$  is the reduced density matrix of the Alice's side, and $\rho^{B}$ is defined similarly. Moreover, $S(\rho )=-\Tr[\rho \,{\log _2}\,\rho ]$ is the von Neumann entropy of the state $\rho$.
Quantum discord at Bob's side reads \cite{ZurekPRL2002}
\begin{equation}\label{QD}
{Q_B}({\rho})=I({\rho})-C_B({\rho}),
\end{equation}
where
\begin{eqnarray}\label{CB}
C_B({\rho })=\mathop {\sup }\limits_{\{ \Pi _k^B\} } \{ S({\rho ^A}) - S({\rho ^A}|\,\{ \Pi _k^B\} )\}.
\end{eqnarray}
Here $S\left( {{\rho ^A}|\,\{ \Pi _k^B\} } \right)=\sum_{k}p_kS(\rho^A_k)$ is the Alice's conditional entropy due to the Bob's measurement.
Equation \eqref{CB} shows that  in order to calculate quantum discord we shall be concerned about the set $\{\Pi _k^B\}$
of all measurements on the Bob's qubit \cite{ZurekPRL2002}. This allows one to extract the most information about the Alice's qubit.

Algorithms to evaluate quantum discord for a general two-qubit state are presented \cite{GirolamiPRA2011,AkhtarshenasIJTP2015}. However,  the optimization problem  requires the solution to a pair of transcendental equations which   involve logarithms of nonlinear quantities \cite{GirolamiPRA2011}. This prevents one to write an analytical expression for the quantum discord even for the simplest case of two-qubit states.  Indeed, quantum discord is analytically computed only for a few families of states including the
Bell-diagonal states \cite{LuoPRA2008,LangPRL2010}, two-qubit $X$ states \cite{AliPRA2010,ChenPRA2010} and two-qubit rank-2 states \cite{ShiJPA2011}.
Using the  Choi-Jamio\l kowski isomorphism, the authors of \cite{WuQINP2015} obtained  the transcendental equations and shown that for a general two-qubit state  they always have  a finite set of universal solutions, however, for some cases such as a subclass of $X$ states, the transcendental equations may offer analytical solutions.

In this paper we use the notion of distinguishability of the Alice's outcomes and  look to those measurements on Bob's
qubit that lead to the most distinguishability of the  Alice's steered states. We show that  such obtained optimum measurement coincides in some cases with the optimum measurement   of Eq. \eqref{CB}.   The correlations  gathered from the most distinguished measurements  give, in general,  a tight upper bound for the  quantum discord.

The paper is organized as follows. In Section II we present  our terminology and provide a brief review for quantum steering ellipsoid. In section III the notion of distinguishability of the Alice's  outcomes is defined and we provide some important classes of states for which the maximum distinguishability can be calculated, analytically. Section IV is devoted to compare our results with quantum discord. The paper is conclude in section V with a brief conclusion.

\section{Framework: Quantum Steering Ellipsoid}
We start from a two-qubit state in the general form as
\begin{equation}\label{Rho-TwoQubit}
\rho  = \frac{1}{4}\left(\Id \otimes \Id+\boldsymbol{x}\cdot\boldsymbol{\sigma}\otimes\Id+ \Id \otimes \boldsymbol{y}\cdot\boldsymbol{\sigma}  + \sum_{i,j=1}^{3}{t_{ij}}{\sigma _i} \otimes {\sigma _j}\right),
\end{equation}
where $\boldsymbol{x}$  and $\boldsymbol{y}$ are Alice and Bob coherence vectors, respectively,  $T=[t_{ij}]$ is the correlation matrix, $\boldsymbol{\sigma}=(\sigma_1,\sigma_2,\sigma_3)$ are the Pauli matrices,  and $\Id$ denotes the unit $2\times 2$ matrix. If Bob performs a projective measurement
\begin{equation}\label{ProjectiveBob}
\Pi_k^B =\frac{1}{2}\left(\Id +\hat{\boldsymbol{n}}_k\cdot \boldsymbol{\sigma}\right), \qquad k = 0,1,
\end{equation}
on his qubit, where $\hat{\boldsymbol{n}}_0=(\sin{\theta}\cos{\phi},\sin{\theta}\sin{\phi},\cos{\theta})^\T=-\hat{\boldsymbol{n}}_1$ and $\T$ denotes the transposition,  the shared bipartite state  collapses to
\begin{equation}\label{RhoAB-PostMeasurement}
\rho= {p_0}\rho _0^A \otimes \Pi _0^B + {p_1}\rho _1^A \otimes \Pi _1^B,
\end{equation}
with
\begin{eqnarray}\label{RhoAk}
\rho_k^A = \frac{1}{2}(\Id + \widetilde{\boldsymbol{x}}_k\cdot\boldsymbol{\sigma}),
\end{eqnarray}
as the post-measurement state of the Alice's side associated with the outcome $k$,  with the corresponding probability
\begin{eqnarray}\label{pk}
p_k =\frac{1}{2}(1 + \boldsymbol{y}\cdot \hat{\boldsymbol{n}}_k).
\end{eqnarray}
Above,  the Alice's post-measurement coherence vector  $\widetilde{\boldsymbol{x}}_k$ is defined by
\begin{equation}\label{xTilde}
\widetilde{\boldsymbol{x}}_k = \frac{{\boldsymbol{x} + T\hat{\boldsymbol{n}}_k}}{1 + \boldsymbol{y}\cdot\hat{\boldsymbol{n}}_k},
\end{equation}
for $k=0,1$.

{\it Canonical states.---}As we mentioned previously, canonical states refer to states for which the Bob's reduced state  is totally mixed, so   $\boldsymbol{y}^{(\can)}=0$.
For such states, it is  easy to construct  Alice's ellipsoid from the above formalism. In this particular case, Alice's post-measurement Bloch vector \eqref{xTilde} reduces to
\begin{equation}\label{xTilde-Canonic}
\widetilde{\boldsymbol{x}}^{(\can)}_k= \boldsymbol{x}^{(\can)}+ T^{(\can)}\hat{\boldsymbol{n}}_k,
\end{equation}
with probability $p^{(\can)}_k=\frac{1}{2}$ for  $k=0,1$.
Since the unit vector $\hat{\boldsymbol{n}}_k$ defines a unit sphere centered at origin, the above equation states  that the set of all points Alice's coherence vector steers to forms an  ellipsoid. This canonical ellipsoid, associated with the canonical state $\rho^{(\can)}$ for which $\boldsymbol{y}^{(can)}=0$,   is  obtained by shrinking and rotating the sphere $\hat{\boldsymbol{n}}_k$  by  matrix $T^{(\can)}$, and then translating it by vector $\boldsymbol{x}^{(\can)}$ \cite{MilnePRL2014-1}.

Interestingly, a canonical state can be obtained from a general state  by local filtering transformation (LFT) \cite{VerstraetePRA2001}. More precisely, starting from a generic two-qubit state $\rho$ with nonzero Bob's coherence vector $\boldsymbol{y}$, one can obtain the canonical state $\rho^{({\can})}$  with $\boldsymbol{y}^{(\can)}=0$  as \cite{MilnePRL2014-1}
\begin{eqnarray}
\rho^{({\can})}&=&\left(\Id \otimes\frac{\Id}{\sqrt{2\rho^B}}\right) \rho \left(\Id \otimes \frac{\Id}{\sqrt{2\rho^B}} \right) \\ \nonumber
&=& \frac{1}{4}\left(\Id \otimes \Id+\boldsymbol{x}^{(\can)}\cdot\boldsymbol{\sigma}\otimes\Id+ \sum_{i,j=1}^{3}{t_{ij}^{(\can)}}{\sigma _i} \otimes {\sigma _j}\right),
\end{eqnarray}
where  the first line  denotes a local filtering transformation on the general state $\rho$.
It is shown that the physicality and separability of states are unchanged under LFT \cite{MilneNJP2014-2}. Furthermore, the Alice's ellipsoid is invariant under LFT on Bob's side,  therefore the LFT makes orbits such that states on the same orbit have equal ellipsoids \cite{MilnePRL2014-1}. In view of this, the canonical states can be considered as the  representatives on the corresponding orbits,  therefore physicality and separability of the states on a general orbit can be determined from the ones  of the canonical states.

\section{Measuring Bob's qubit with the most disruptive Alice's qubit}
As we mentioned already, in order to calculate  quantum discord we shall be concerned about the set  $\{ \Pi _k^B\}$ of all measurements on the Bob's qubit. This allows one to extract the most information about the Alice's qubit \cite{ZurekPRL2002} and that at the same time disturbs least the overall quantum state $\rho$.   This corresponds also to finding measurements that  maximize  Eq. \eqref{CB}. When the  results of the measurement  are not recorded, the  measurement on Bob's qubit does not disturb Alice's state $\rho^A$. However, corresponding to the measurement outcomes,  Alice's state steers to  some states $\rho^A_k$ in her ellipsoid  following the route of Eqs. \eqref{RhoAk} and \eqref{pk}. The ability to extract   information about Alice's qubit by measuring Bob's qubit comes from  correlations shared between them  and this is, in general,  accompanied by disrupting the Alice's outcome states.
Now the question arises:   To what extent does the extraction of the most information about Alice's qubit disturb  her outcomes?  To address this question let us consider the Bell-diagonal states, i.e. states described by Eq. \eqref{Rho-TwoQubit} with  $\boldsymbol{x}=\boldsymbol{y}=0$ and $T=\diag\{t_1,t_2,t_3\}$. A comparison between  the Alice's conditional entropy and the Euclidean distance of the Alice's steered states shows that these functions behave, in general, oppositely under local measurement on the Bob's side (see Fig. \ref{Fig1}). In particular, we  observe that the minimum of the conditional entropy coincides with the maximum of the Euclidean distance of  two steered states.  It seems therefore that   extracting the most information about Alice's qubit can be associated with the most disturbing her outcomes states.
\begin{figure}
\mbox{\includegraphics[scale=0.95]{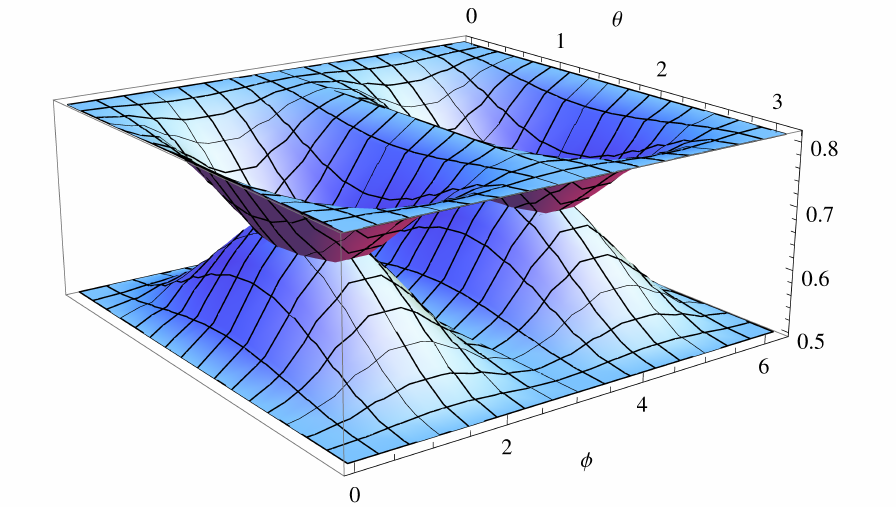}}
\caption{(Color online)  The conditional entropy (upper surface) and the Euclidean distance (lower surface) of  two steered states of the Alice's qubit  as a function of the  Bob's measurement parameters $(\theta,\phi)$ for a Bell-diagonal state with $(t_1,t_2,t_3)=(-0.5,0.7,0.5)$. The figure shows that these functions have opposite behaviour  under local measurement on the Bob's qubit.  }
\label{Fig1}
\end{figure}
Motivated by the above observation, in what follows  we are looking to those  measurements on Bob's qubit that cause the most disturbance in the Alice's post-measurements states. To this aim we use the {\em trace distance} as a measure of quantum distinguishability between two outcomes \cite{NielsenBook2000}
\begin{eqnarray}\label{TrDistance}
D(\rho^A_0,\rho^A_1)=\Tr{|\rho^A_0-\rho^A_1|}.
\end{eqnarray}
This,  in turns,   reduces simply to the Euclidian distance $D(\widetilde{\boldsymbol{x}}_0,\widetilde{\boldsymbol{x}}_1)$ between Bloch vectors of the two post-measurement states $\rho_0^A$ and $\rho_1^A$. For the squared distance we find
\begin{eqnarray}\label{D2}
D^2(\widetilde{\boldsymbol{x}}_0,\widetilde{\boldsymbol{x}}_1)=|\widetilde{\boldsymbol{x}}_0-\widetilde{\boldsymbol{x}}_1|^2 =\frac{4\hat{\boldsymbol{n}}^\T M \hat{\boldsymbol{n}}}{(1-\hat{\boldsymbol{n}}^\T Y \hat{\boldsymbol{n}})^2},
\end{eqnarray}
where $Y=\boldsymbol{y} \boldsymbol{y}^\T$ and $M=m^\T m$ with $m=(T-\boldsymbol{x}\boldsymbol{y}^\T)$. Maximum distinguishability corresponds therefore to the  maximum distance given by
\begin{eqnarray}\label{OptDis}
D^2_{\max}(\widetilde{\boldsymbol{x}}_0,\widetilde{\boldsymbol{x}}_1)=\max_{\hat{\boldsymbol{n}}}\left[\frac{4\hat{\boldsymbol{n}}^\T M \hat{\boldsymbol{n}}}{(1-\hat{\boldsymbol{n}}^\T Y \hat{\boldsymbol{n}})^2}\right],
\end{eqnarray}
where maximum is taken over all unit vectors $\hat{\boldsymbol{n}}\in\mathbb{R}^3$.
Before we proceed further to find conditions under which $D^2(\widetilde{\boldsymbol{x}}_0,\widetilde{\boldsymbol{x}}_1)$ is maximize, let us turn our attention  on some particular cases for which the maximum is  obtained analytically without any need for rigorous optimization.
\\
{\it (i) Canonical states $\boldsymbol{y}=0$.---}For the important class of canonical states for which the Bob's coherence vector is zero, the optimum measurement leading  to the maximum distance between Bloch vectors of the post-measurement states  is nothing but the eigenvector of $T$  corresponding to its largest eigenvalue. Therefore in this case we have $D^2_{\max}(\widetilde{\boldsymbol{x}}_0,\widetilde{\boldsymbol{x}}_1)=4\max\{t_1^2,t_2^2,t_3^2\}$.
The  Bell-diagonal states, for which  the Alice's coherence vector also vanishes, are an important subclass of canonical states.
\\
{\it(ii) States for which $\boldsymbol{y}$ is an eigenvector corresponding to the largest eigenvalue of $M$.---}In this case maximum of the enumerate happens in the direction of coherence vector of the part $B$, i.e. $\max_{\hat{\boldsymbol{n}}}\hat{\boldsymbol{n}}^\T M \hat{\boldsymbol{n}}=\boldsymbol{y}^\T M \boldsymbol{y}/y^2$. For such states we get $D^2_{\max}(\widetilde{\boldsymbol{x}}_0,\widetilde{\boldsymbol{x}}_1)=\left[\frac{4\boldsymbol{y}^\T M \boldsymbol{y}}{y^2(1-y^2)^2}\right]$.
\\
{\it (iii) $X$ states.---}The important class of $X$ states is defined by $\boldsymbol{x}=\left(\begin{array}{ccc}0 & 0 & x \end{array}\right)^\T$, $\boldsymbol{y}=\left(\begin{array}{ccc}0 & 0 & y \end{array}\right)^\T$, and $T=\diag\{t_1,t_2,t_3\}$. In this case  $M$ is also a diagonal matrix given by  $M=\diag\{M_1,M_2,M_3\}$ with $M_1=t_1^2$, $M_2=t_2^2$ and $M_3=(t_3-xy)^2$. In what follows we assume that $|t_1|\ge |t_2|$ ($|t_1|\le |t_2|$ can be obtained just by replacing $1\rightarrow 2$ and $x\rightarrow y$). In this case we find the following results.
\begin{enumerate}
\item $M_1\le M_3$. For such case we get $D_{\max}^2(\widetilde{\boldsymbol{x}}_0,\widetilde{\boldsymbol{x}}_1)=\frac{4M_3}{(1-y^2)^2}$ with $\sigma_z$ as the optimal  measurement.
\item $M_1\ge M_3$. In this case the optimal  measurement is defined by $(\hat{n}^\ast_1)^2=1-(\hat{n}^\ast_3)^2$, $\hat{n}^\ast_2=0$, and  $(\hat{n}^\ast_3)^2=\frac{2M_1y^2-(M_1-M_3)}{(M_1-M_3)y^2}$  if
    \begin{equation}
    \frac{M_1-M_3}{2M_1}\le y^2 \le \frac{M_1-M_3}{M_1+M_3},
    \end{equation}
    or equivalently
    \begin{equation}
    \left(\frac{2M_3}{M_1+M_3}\right)^2\le (1-y^2)^2 \le \left(\frac{M_1+M_3}{2M_1}\right)^2.
\label{19}
    \end{equation}
    On the other hand, the optimal measurement is $\sigma_z$, i.e. $D_{\max}^2(\widetilde{\boldsymbol{x}}_0,\widetilde{\boldsymbol{x}}_1)=\frac{4M_3}{(1-y^2)^2}$, if
    \begin{equation}
    \left(\frac{2M_3}{M_1+M_3}\right)^2\ge (1-y^2)^2,
\label{20}
    \end{equation}
    and it is  $\sigma_x$, i.e. $D_{\max}^2(\widetilde{\boldsymbol{x}}_0,\widetilde{\boldsymbol{x}}_1)=4M_1$, if
    \begin{equation}
    (1-y^2)^2 \ge \left(\frac{M_1+M_3}{2M_1}\right)^2.
\label{21}
 \end{equation}
\end{enumerate}

Now,  after giving the maximum distance for some particular  classes of states without rigorous optimization, we  provide in what follows an analytical procedure for optimization of Eq. \eqref{OptDis}. In order to determine the maximum distance, we have to calculate its derivatives with respect to $\theta$ and $\phi$. For derivative with respect to  $\theta$ we get
\begin{eqnarray}
\frac{\partial D^2}{\partial \theta}=\frac{8{\hat{\boldsymbol{n}}}_{,\theta}^{\T} \mathcal{M}({\hat{\boldsymbol{n}}}) {\hat{\boldsymbol{n}}}}{(1-\hat{\boldsymbol{n}}^\T Y \hat{\boldsymbol{n}})^2},
\end{eqnarray}
where  $\mathcal{M}({\hat{\boldsymbol{n}}})$ is a ${\hat{\boldsymbol{n}}}$-dependent  symmetric matrix given by
\begin{equation}
\mathcal{M}({\hat{\boldsymbol{n}}})=M+\frac{2{\hat{\boldsymbol{n}}}^\T M {\hat{\boldsymbol{n}}}}{(1-\hat{\boldsymbol{n}}^\T Y \hat{\boldsymbol{n}})}Y,
\end{equation}
and the unit vector ${\hat{\boldsymbol{n}}}_{,\theta}$ is defined by
\begin{equation}
{\hat{\boldsymbol{n}}}_{,\theta}=\frac{\partial{\hat{\boldsymbol{n}}}}{\partial \theta}=(\cos{\theta}\cos{\phi},\cos{\theta}\sin{\phi},-\sin{\theta})^{\T}.
\end{equation}
Evidently ${\hat{\boldsymbol{n}}}_{,\theta}\cdot {\hat{\boldsymbol{n}}}=0$.
By defining the nonunit vector $\tilde{\boldsymbol{n}}_{,\phi}$ by
\begin{equation}
{\tilde {\boldsymbol{n}}}_{,\phi}=\frac{\partial{\hat{\boldsymbol{n}}}}{\partial \phi}=(-\sin{\theta}\sin{\phi},\sin{\theta}\cos{\phi},0)^{\T},
\end{equation}
 orthogonal to both ${\hat{\boldsymbol{n}}}$ and ${\hat{\boldsymbol{n}}}_{,\theta}$, we  get a similar equation for the derivative of the distance  with respect to $\phi$, but now ${\hat{\boldsymbol{n}}}_{,\theta}$ is replaced by ${\tilde{ \boldsymbol{n}}}_{,\phi}$.
Excluding the case $y=1$ which happens if and only if the overall state is pure, we find the following relation  for the stationary condition $\frac{\partial D^2}{\partial \theta}=\frac{\partial D^2}{\partial \phi}=0$,
\begin{equation}\label{Stationary}
{\hat{\boldsymbol{n}}}_{\perp}^{\T} \;\mathcal{M}({\hat{\boldsymbol{n}}}) \;{\hat{\boldsymbol{n}}}=0,
\end{equation}
where $\hat{\boldsymbol{n}}_{\perp}$ is any vector perpendicular to $\hat{\boldsymbol{n}}$, i.e. $\hat{\boldsymbol{n}}_{\perp}\cdot\hat{\boldsymbol{n}}=0$.
This implies that the stationary points are achieved if  and only if ${\hat{\boldsymbol{n}}}$ be an eigenvector of $\mathcal{M}({\hat{\boldsymbol{n}}})$. Note that knowing   the extremum points of the distance is not enough to
establish its maximum, and we are required a  further investigation of  the distance over all
extremum points to get  the maximum one. Although   condition  \eqref{Stationary} does not provide an easy  solution for the maximum of the distance, due to the dependence of the symmetric matrix $\mathcal{M}({\hat{\boldsymbol{n}}})$ on the unknown direction  ${\hat{\boldsymbol{n}}}$,   it provides still a simple condition to evaluate   the stationary points numerically.   Not surprisingly, the above stationary condition is fulfilled for  the special classes of states for which we have already obtained the maximum distance without rigorous optimization.

From the discussion given at the beginning of this section, two questions are  being raised.  The first one is that, is there any relation between the optimum measurement associated with the maximum distinguishability of the Alice's outcomes with the one that allows one to extract the most information about the Alice's qubit? We demonstrate in the following section that this is, indeed, the case. To do so, we provide some examples for which these two optimum measurements coincide exactly. The second question is that,  when the optimum measurement of the  maximum distinguished-outcomes process  differs from the most information-gathering one, whether the former  can be used to find a tight and faithful upper bound on the quantum discord?   We will address these questions in the next section.

\section{Maximum distinguished-outcomes measurement versus  the most information-gathering one }
Suppose  Bob performs a measurement on his qubit in the direction $\hat{\boldsymbol{n}}^\ast$  which fulfills the maximum distinguishability condition \eqref{OptDis}. Using this in the definition of quantum discord we find that
\begin{equation}\label{QQast}
Q_B(\rho)\le Q_B^\ast(\rho),
\end{equation}
where $Q_B(\rho)$ is the quantum discord of $\rho$, Eq. \eqref{QD}, and $Q_B^\ast(\rho)$ is
its upper bound  defined by
\begin{eqnarray}\nonumber
{Q_B^\ast}({\rho})&=&S({\rho ^B}) - S({\rho ^{AB}})+  S({\rho ^A}|\,\{ \Pi^{\ast B}_k\}) \\ \label{Qast}
&=& [h_2(\vec{q})-h_4(\vec{\lambda})]+ [h_4(\vec{w})-h_2(\vec{p})].
\end{eqnarray}
Above,    $\Pi^{\ast B}_k=\frac{1}{2}\left(\Id +\hat{\boldsymbol{n}}^\ast_k\cdot \boldsymbol{\sigma}\right)$ for $k = 0,1$ is the optimum measurement  that maximizes $ D^2(\widetilde{\boldsymbol{x}}_0,\widetilde{\boldsymbol{x}}_1)$. Moreover,
$h_m(\vec{x})$ stands for  the   Shannon entropy  of  a probability vector of length $m$. Also $\vec{\lambda}$ and $\vec{q}$ denotes the probability vectors constructed from  the eigenvalues of $\rho$ and $\rho^B$, respectively, and $\vec{w}$ and $\vec{p}$ are two probability vectors of length 4 and 2, respectively, given by \cite{AkhtarshenasIJTP2015}
\begin{eqnarray}\label{wkl}
w_{k,l}&=&\frac{1}{4}\left\{{1+(-)^{k}\boldsymbol{y}\cdot \hat{\boldsymbol{n}}^\ast+(-)^{l} \left|\boldsymbol{x}+(-)^{k}T \hat{\boldsymbol{n}}^\ast\right|}\right\}, \\ \label{pk}
p_{k}&=&\frac{1}{2}\left\{1 +(-)^k \boldsymbol{y}\cdot \hat{\boldsymbol{n}}^\ast\right\},
\end{eqnarray}
for $k,l\in \{0,1\}$.
The following lemma shows that  the above upper bound is faithful in a sense that it vanishes if and only if the bounded quantity vanishes.

\begin{lemma}\label{Faithful}
$Q^\ast_B(\rho)=0$ if and only if $Q_B(\rho)=0$.
\end{lemma}
\begin{proof}
The sufficient condition is a simple consequence of Eq. \eqref{QQast}. To prove the necessary condition, let $\rho$ be a zero-discord on the  Bob's side.   A two-qubit state  has zero discord on  Bob's side   if and only if either (i) $T = 0$, or (ii) $\rank(T )=1$  and $\boldsymbol{y}$ belongs to the range of $T$ \cite{LuPRA2011,Saman}.
We need therefore to prove that   both   cases lead to $Q_B^\ast(\rho)=0$.

(i) If $T=0$, we have from Eq. \eqref{OptDis}
\begin{eqnarray} \nonumber
D^2_{\max}(\widetilde{\boldsymbol{x}}_0,\widetilde{\boldsymbol{x}}_1)=\max_{\hat{\boldsymbol{n}}}\left[\frac{4x^2(\hat{\boldsymbol{n}}\cdot \boldsymbol{y})^2}{(1-(\hat{\boldsymbol{n}}\cdot \boldsymbol{y})^2)^2}\right],
\end{eqnarray}
which takes its maximum value for $\hat{\boldsymbol{n}}^\ast=\hat{\boldsymbol{y}}=\boldsymbol{y}/|\boldsymbol{y}|$. On the other hand, in this case, simple calculation shows that  eigenvalues of $\rho$ and $\rho^B$ are given by
\begin{eqnarray*}
 \lambda_{k,l}&=&\frac{1}{4}\left\{1+(-)^{k}y+(-)^{l} x\right\},\quad q_{k}=\frac{1}{2}\left\{1+(-)^k y\right\},
\end{eqnarray*}
respectively ($k,l\in \{0,1\}$).
Using these and  putting $T=0$ in Eqs. \eqref{wkl} and \eqref{pk}, we find from  Eq. \eqref{Qast} that    $\hat{\boldsymbol{n}}^\ast=\hat{\boldsymbol{y}}$ gives  ${Q_B^\ast}({\rho})=0$.

(ii) For the second case,  i.e. when  $\rank(T )=1$  and $\boldsymbol{y}$ belongs to the range of $T$,  without any loss of generality  we assume that  $\boldsymbol{y}$ and $T$ have the   form $\boldsymbol{y}=y\hat{\boldsymbol{k}}$ and $T=t\hat{\boldsymbol{k}}\hat{\boldsymbol{k}}^\T$,
respectively.  In this case Eq. \eqref{OptDis} leads to
\begin{eqnarray*}
D^2_{\max}(\widetilde{\boldsymbol{x}}_0,\widetilde{\boldsymbol{x}}_1)=\max_{\hat{\boldsymbol{n}}} \left[\frac{\left( {{t}^2 + {x^2}{\kern 1pt} {\kern 1pt} {y}^2 - 2{\kern 1pt} {\kern 1pt} {t}{\kern 1pt} {\kern 1pt} {x}{\kern 1pt} {\kern 1pt} {y}} \right) (\hat{\boldsymbol{n}}\cdot \hat{\boldsymbol{k}})^2}{(1 - {y}^4(\hat{\boldsymbol{n}}\cdot \hat{\boldsymbol{k}})^2 )^2}\right],
\end{eqnarray*}
which takes its maximum value for $\hat{\boldsymbol{n}}^\ast=\hat{\boldsymbol{k}}$. For  such states we have
\begin{eqnarray*}
\lambda_{k,l}&=&\frac{1}{4}\left\{1+(-)^{k}y+(-)^{l} \left|\boldsymbol{x}+(-)^{k}T\hat{\boldsymbol{k}}\right|\right\}, \\
q_{k}&=&\frac{1}{2}\left\{1+(-)^k y\right\},
\end{eqnarray*}
for  eigenvalues of $\rho$ and $\rho^B$, respectively. Moreover,  $w_{k,l}$ and $p_{k}$ are given by  Eqs. \eqref{wkl} and \eqref{pk} with $\boldsymbol{y}=y\hat{\boldsymbol{k}}$ and $T=t\hat{\boldsymbol{k}}\hat{\boldsymbol{k}}^\T$. A simple investigation shows that ${Q_B^\ast}({\rho})=0$  for  $\hat{\boldsymbol{n}}^\ast=\hat{\boldsymbol{k}}$. This completes the proof.
\end{proof}

In what follows we show that the above upper bound is tight in a sense that in  more situations the equality is saturated.  To this aim  we consider  states that we have considered in the last subsection.

{\it (i) Canonical states $\boldsymbol{y}=0$.---}There is no complete solution to the quantum discord of the canonical states, although their geometry and optimization formula are simpler than the general states. Without losing generality we assume $\left| {{t_1}} \right| \ge \left| {{t_2}} \right|$ and then do measurement along the greater semi-axis between $\hat{\boldsymbol{ i}}$ and $\hat{\boldsymbol{ k}}$. We do this and plot the results versus the quantum discord in Fig. \ref{Fig2}  for more than 20000 random states. There are many points on the bisector line showing that the optimized direction is very near to the direction of  our upper bound. Moreover,  non-exact results are not too far from quantum discord and distribution of points near the bisector line shows that the upper bound is very near to the quantum discord.
Canonical states with $T^\T\boldsymbol{x}=0$  have been solved analytically in \cite{AkhtarshenasIJTP2015} and it is easy to see that for this subclass we have $Q_B(\rho)=Q_B^\ast(\rho)$.
\begin{figure}
\mbox{\includegraphics[scale=0.75]{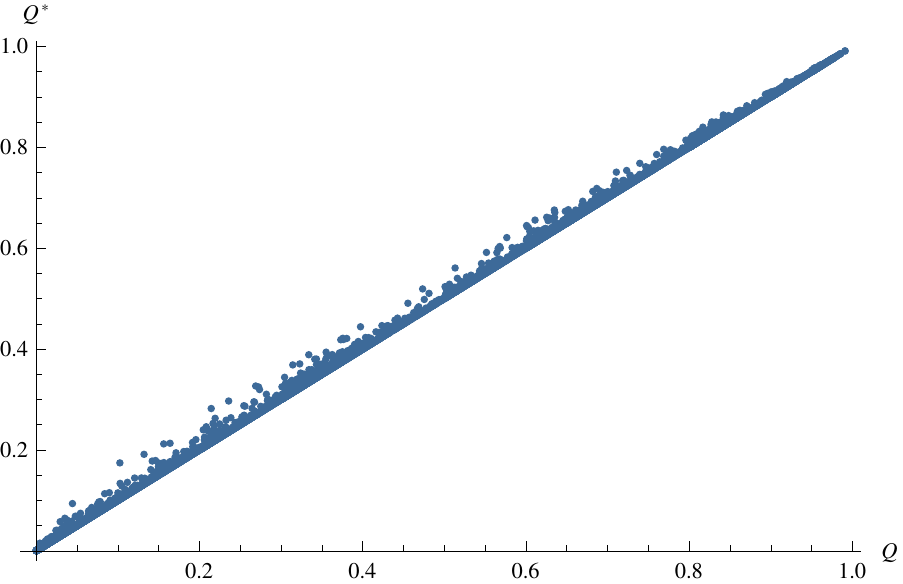}}
\caption{(Color online) The upper bound $Q^\ast(\rho)$ \textit{vs.} quantum discord $Q(\rho)$ for 20000 random canonical states. Points on the bisector line belong to states for which the upper bound is equal to QD. }
\label{Fig2}
\end{figure}

{\it(ii) States for which $\boldsymbol{y}$ is an eigenvector corresponding to the largest eigenvalue of $M$.---}In this case there are  some classes of states for which there exist a good agreement between $Q_B(\rho)$ and $Q^\ast_B(\rho)$. Consider states with
 \begin{eqnarray}
\boldsymbol{x} = x\hat{\boldsymbol{ k}}, \qquad \boldsymbol{y} = y\hat{\boldsymbol{ i}}, \qquad T = \diag\{ {t_1},{t_2},0\}.
 \end{eqnarray}
In this case  $M = \diag\{ {t_1}^2 + {x^2}{y^2},{t_2}^2,0\} $,  and when ${t_1}^2 + {x^2}{y^2} \ge {t_2}^2$ both $Q_B(\rho)$ and $Q_B^\ast(\rho)$ are obtained by measurement  along $\boldsymbol{y}$. In Fig. \ref{Fig3} we have plotted $Q_B^\ast(\rho)$ versus $Q_B(\rho)$  for more than 3000 random states of this category.
\begin{figure}
\mbox{\includegraphics[scale=0.75]{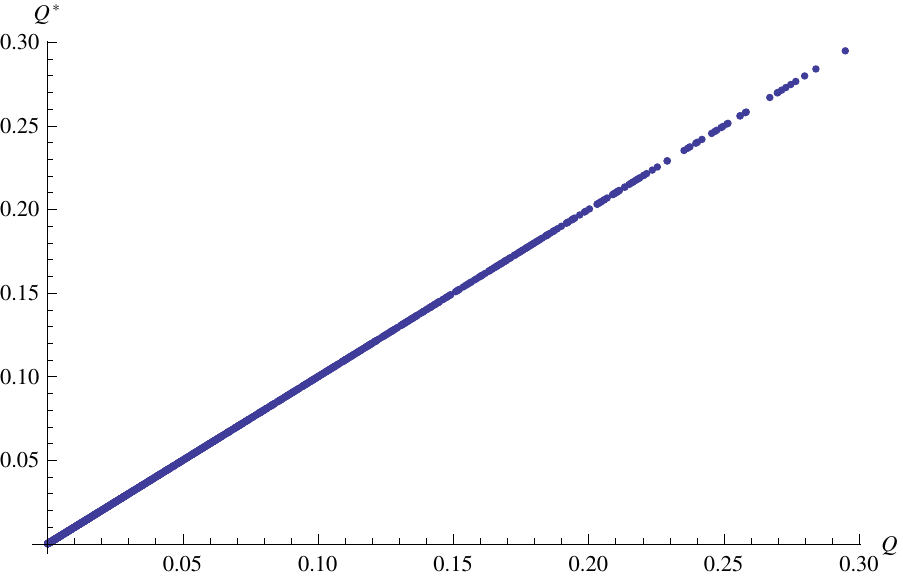}}
\caption{(Color online) The upper bound $Q^\ast(\rho)$ \textit{vs.} quantum discord $Q(\rho)$ for 3000 random states for which $\boldsymbol{y}$ is an eigenvector corresponding to the largest eigenvalue of $M$. Results have high accuracy with $RE<10^{-6}$.}
\label{Fig3}
\end{figure}

{\it (iii) $X$ states.---}For $X$-states we consider the following classes separately.
\begin{enumerate}
\item $M_1\le M_3$. In this case $Q_B^\ast(\rho)$ is very near to $Q_B(\rho)$ and the relative error is less than $10^{-6}$. For 10000 random states of this category any point lies on the bisector line (Fig. \ref{Fig4}).

\begin{figure}
\mbox{\includegraphics[scale=0.75]{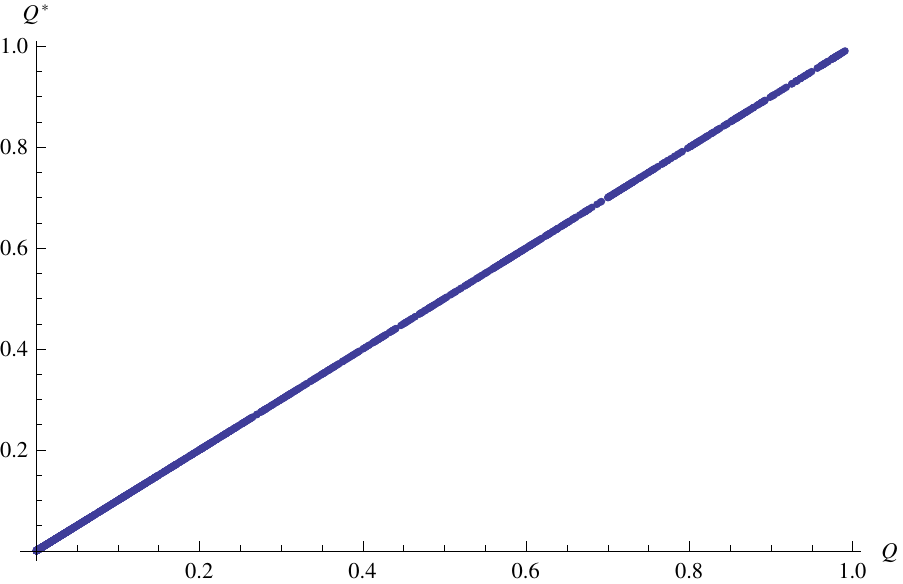}}
\caption{(Color online) The upper bound $Q^\ast(\rho)$ \textit{vs.} quantum discord $Q(\rho)$ for 10000 random  $X$ states satisfying $M_1\le M_3$. Results have high accuracy with $RE < 10^{-6}$.}
\label{Fig4}
\end{figure}

\item $M_1\ge M_3$.   In Figs . \ref{Fig5} and \ref{Fig6}  we plot $Q_B^\ast(\rho)$ \textit{vs.} $Q_B(\rho)$ for 20000 random states satisfying one of the Eqs. \eqref{20} and \eqref{21}, and 5000 random states satisfying Eq. \eqref{19}, respectively.
\end{enumerate}

\begin{figure}
\mbox{\includegraphics[scale=0.75]{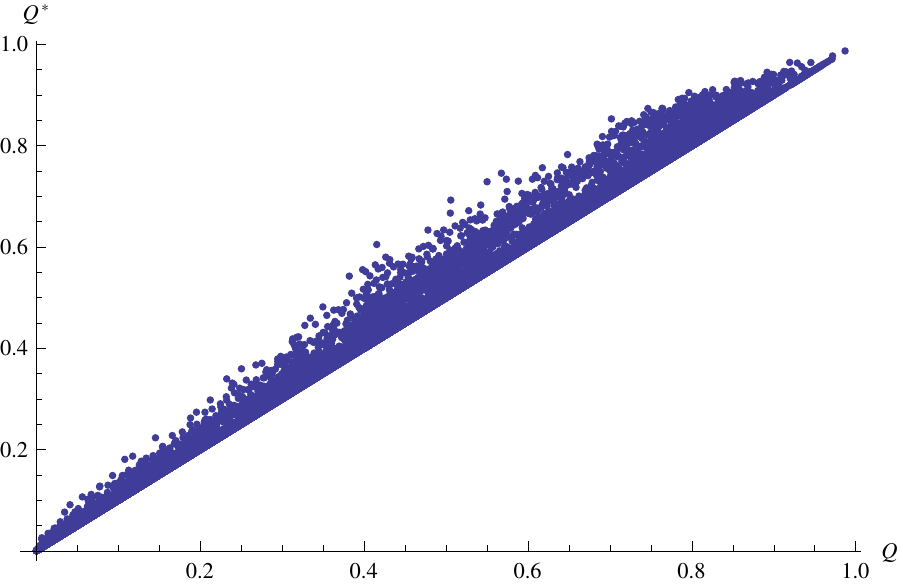}}
\caption{(Color online) The upper bound $Q_B^\ast(\rho)$ \textit{vs.} quantum discord $Q_B(\rho)$ for 20000 random  $X$ states satisfying $M_1\ge M_3$ and one of the Eqs. \eqref{20} and \eqref{21}.  }
\label{Fig5}
\end{figure}

\begin{figure}
\mbox{\includegraphics[scale=0.75]{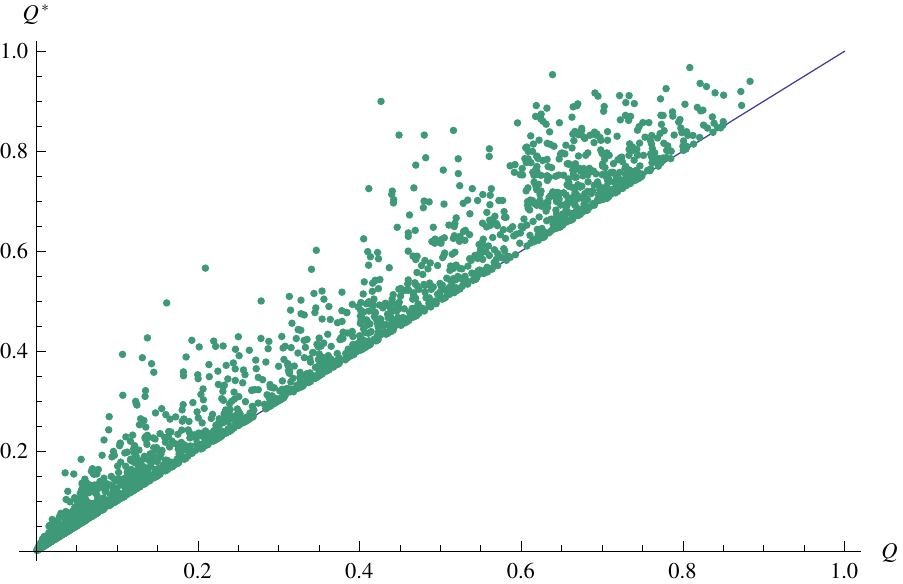}}
\caption{(Color online) The upper bound $Q_B^\ast(\rho)$ \textit{vs.} quantum discord $Q_B(\rho)$ for 5000 random  $X$ states satisfying $M_1\ge M_3$ and Eq. \eqref{19}.  }
\label{Fig6}
\end{figure}

\subsection{$Q_B^\ast(\rho)$ as a tight upper bound}

Now we proceed to employ $Q_B^\ast(\rho)$ as an upper bound and check if it is a tight one. Here we focus on a two parameters state as \cite{AlqasimiPRA2011}
\begin{equation}
\rho = \frac{1}{2}\left( {\begin{array}{*{20}{c}}
a&0&0&a\\
0&{1 - a - b}&0&0\\
0&0&{1 - a + b}&0\\
a&0&0&a
\end{array}} \right)\,,
\label{27}
\end{equation}
where $0 \le a \le 1$ and $a-1 \le b \le 1-a$. The quantum discord of this state is \cite{AlqasimiPRA2011}
\begin{equation}
Q_B(\rho) =\min \{ a,q\},
\end{equation}
where
\begin{eqnarray}
q &=& \frac{a}{2}\,\log_2[\frac{{4\,{a^2}}}{{{{\left( {1 - a} \right)}^2} - {b^2}}}] - \frac{b}{2}\,log_2[\frac{{(1 + b)(1 - a - b)}}{{(1 - b)(1 - a + b)}}]
 \\ \nonumber
 &+& \frac{1}{2}\,\log_2[\frac{{4({{(1 - a)}^2} - {b^2})}}{{(1 - {b^2})(1 - {a^2} - {b^2})}}] - \frac{{\sqrt {{a^2} + {b^2}} }}{2}\,\log_2[\frac{{1 + \sqrt {{a^2} + {b^2}} }}{{1 - \sqrt {{a^2} + {b^2}} }}].
\end{eqnarray}
Here $a$ and $q$ are obtained by measurements $\hat{\boldsymbol{n}}=\hat {\boldsymbol{i}}$ and $\hat{\boldsymbol{n}}=\hat{\boldsymbol{ k}}$, respectively.
Marginal coherence vectors and the correlation matrix of this state are given by
\begin{equation}
\boldsymbol{x} =  - \boldsymbol{y} = \left( {\begin{array}{c}
0\\
0\\
{ - b}
\end{array}} \right),
\qquad T = \left( {\begin{array}{ccc}
a & 0 & 0\\ 0 &  - a & 0 \\ 0 & 0 & 2a - 1
\end{array}} \right).
\end{equation}
On the other hand, the maximal distance of this state is
\begin{equation}
D^2_{\max}=\max \left\{a^2,\frac{{{{\left( {(2a - 1) + {b^2}} \right)}^2}}}{{{{\left( {1 - {b^2}} \right)}^2}}}\right\},
\end{equation}
where  $a^2$ and $\frac{{{{\left( {(2a - 1) + {b^2}} \right)}^2}}}{{{{\left( {1 - {b^2}} \right)}^2}}}$ are  obtained by measurements $\hat{\boldsymbol{n}}=\hat {\boldsymbol{i}}$ and $\hat{\boldsymbol{n}}=\hat{\boldsymbol{ k}}$, respectively.

Evidently, for $b=0$ we have  $Q_B(\rho) =  Q_B^\ast(\rho)$. In Fig. \ref{Fig7}  we plot $Q_B(\rho)$ and $Q_B^\ast(\rho)$ as a function of $a$ for two cases $b=0.3$ and $b=0.7$, respectively. Except at a very small  interval, we have   $Q_B(\rho) =  Q_B^\ast(\rho)$. A comparison of these figures reveals  that  as $b$ increases the amounts of QD decreases and also the interval in which $Q_B(\rho) \neq  Q_B^\ast(\rho)$, grows up. Therefore, we observe that in high discordant states $Q_B^\ast(\rho)$ is more precise.

\begin{figure}
\centering
\subfigure[]{\includegraphics[width=0.4\textwidth]{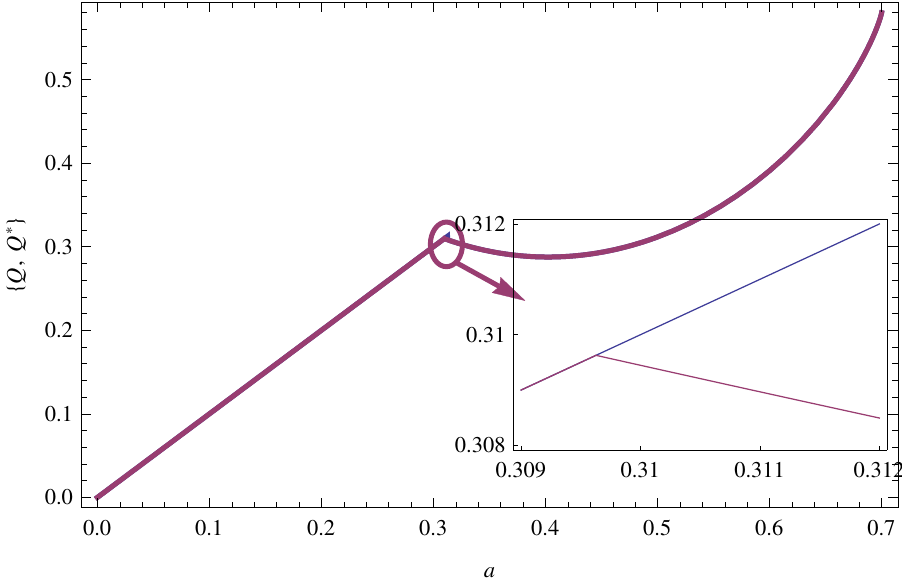}}
\subfigure[]{\includegraphics[width=0.4\textwidth]{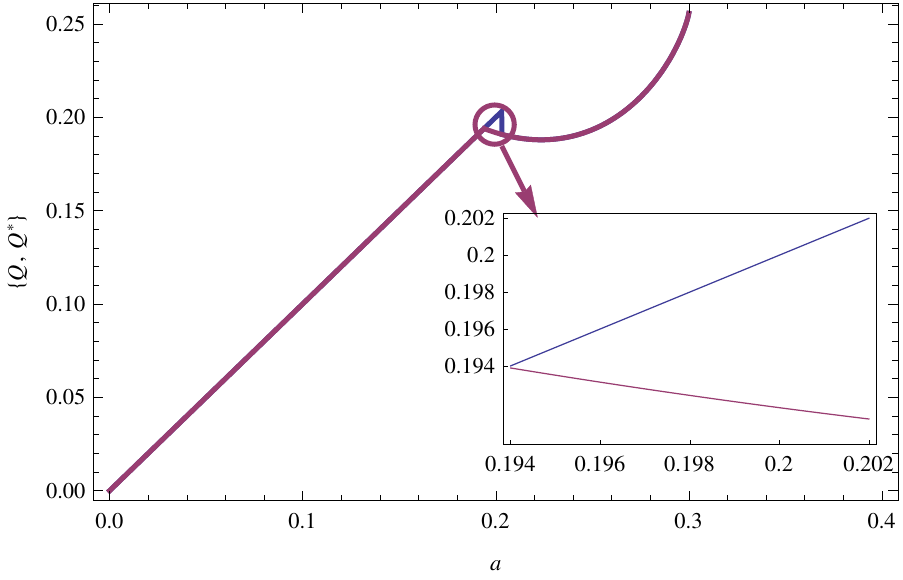}}
\caption{ (Color online) $Q_B(\rho)$ and $Q_B^\ast(\rho)$ as  a function of $a$ for (a)  $b=0.3$ and (b)  $b=0.7$. The insets show behaviour in   the rigion near  $Q_B(\rho) \neq  Q_B^\ast(\rho)$.}
\label{Fig7}
\end{figure}

\section{Conclusion}

Here we have defined   $Q^\ast(\rho)$ as the correlation that Bob can extract about Alice's qubit by means of the most distinguishable measurements, i.e.  measurements that Bob  steers Alice to the most distinguishable states.  For some  classes of states, we have shown  that this quantity is equal to the quantum discord  $Q(\rho)$. Although    $Q^\ast(\rho)$   may contain some classical correlations,   the amount of classical correlations is  not so much in particular  for  high discordant states. The presented quantity provides  a faithful and  tight upper bound for the quantum discord. Marginal states  at  high discordant states  have high mixedness and so they are near to the  Bell-diagonal states, for which $Q^\ast(\rho)$ coincides exactly with  $Q(\rho)$.

The significance of our method  comes from two facts: (i) the tightness of the provided bound and (ii) the physical interpretation of this bound. As we mentioned above, the provided upper bound is faithful and tight, meaning that the bound  vanishes if and only if the
bounded quantity vanishes.  This, in turn, indicates that a nonzero value for the upper bound is a guarantee for the nonzero quantum discord, a fact  that is not valid, in general,  for an arbitrary  upper bound. In other words, for zero discord states with possible classical correlations, the most distinguishable measurement washes out all  classical correlations.  On the other hand, the physical interpretation of our method is related to its relevance to the notion of  quantum distinguishability. Actually a look at measure of distinguishability of two states, given by  Eq. \eqref{TrDistance} for  outcomes of the Alice's side when Bob performs  a von Neumann measurement on his particle, shows that this measure is closely related  to the notion of  minimum-error  probability of discrimination of two states for equal a priori probability \cite{NielsenBook2000}.

The generalization of the above method is not straightforward as there are  not  well known geometries like Bloch sphere and quantum steering ellipsoid for arbitrary bipartite systems. For a general bipartite state  with arbitrary dimension for Bob's particle, when  Bob performs POVM measurement on his particle with $n$ outcomes, the state of the Alice's side steers to $\rho^A_k$ with probability $p_k$  corresponding to each outcome $k=1,\cdots,n$. Following the route of two-qubit system, we left therefore with the problem of finding the best possible measurement on the Bob's side with the outputs that are most distinguishable on the Alice's side. This, however, is not an easy task to treat in general and further study on the subject is under our investigation.

\acknowledgments This work was supported by Ferdowsi University of Mashhad under Grant No.  3/38668 (1394/06/31).

\end{document}